\documentclass[11pt]{llncs}

\newcommand{\myifthen}[2]{#1}

\usepackage{makeidx}
\usepackage{amsfonts,amssymb,amsmath}
\usepackage{epsfig,graphics, bbm, picins}
\usepackage{enumerate,mathbbol}

\myifthen{\usepackage[letterpaper,margin=1in]{geometry}}{}

\newcommand{\IR}{\mathbb{R}}

\newcommand{\st}{\ensuremath{\mbox{s.t.}}}

\newcommand{\ignore}[1]{}
\newcommand{\0}{\mathbb{0}}
\newcommand{\1}{\mathbb{1}}
\newcommand{\T}{\mathcal{T}}

\newcommand{\core}{\ensuremath{\mathcal{C}}}
\newcommand{\prekernel}{\ensuremath{\mathcal{P}}}
\newcommand{\nwb}{NB}
\begin{document}
\pagestyle{headings}

\mainmatter

\title{Network Bargaining: Using Approximate Blocking Sets to Stabilize
 Unstable Instances }
\titlerunning{ Network Bargaining: Using Approximate Blocking Sets to Stabilize Unstable Instances }

\author{
Jochen K\"onemann\inst{1}
\and
Kate Larson\inst{2}
\and
David Steiner\inst{2}
}
\authorrunning{K\"onemann et al.}
\tocauthor{Jochen K\"onemann, Kate Larson, David Steiner}

\institute{Department of Combinatorics and Optimization,
University of Waterloo, Waterloo, Ontario N2L 3G1, Canada\\
\email{jochen@uwaterloo.ca}
\and
Cheriton School of Computer Science,
University of Waterloo, Waterloo, Ontario N2L 3G1, Canada\\
\email{\{klarson,dasteine\}@uwaterloo.ca}
}

\maketitle

\begin{abstract}
  We study a network extension to the Nash bargaining game, as
  introduced by Kleinberg and Tardos~\cite{KT08}, where the set of
  players corresponds to vertices in a graph $G=(V,E)$ and each edge
  $ij\in E$ represents a possible deal between players $i$ and $j$.
  We reformulate the problem as a cooperative game and study the
  following question: \emph{Given a game with an empty core (i.e. an
    unstable game) is it possible, through minimal changes in the
    underlying network, to stabilize the game?}  We show that by
  removing edges in the network that belong to a \emph{blocking set}
  we can find a stable solution in polynomial time. This motivates the
  problem of finding small blocking sets. While it has been previously
  shown that finding the smallest blocking set is
  NP-hard~\cite{BKP10}, we show that it is possible to efficiently
  find approximate blocking sets in sparse graphs.

\end{abstract}

\section{Introduction}
%%%%%%%%%%%%%

In the classical {\em Nash bargaining} game~\cite{Na50}, two players
seek a mutually acceptable agreement on how to split a dollar.  If no
such agreement can be found, each player $i$ receives her {\em
  alternative} $\alpha_i$.  Nash's solution postulates, that in an
equilibrium, each player $i$ receives her alternative $\alpha_i$ plus
half of the surplus $1-\alpha_1-\alpha_2$ (if
$\alpha_1+\alpha_2>1$ then no mutually acceptable agreement can be
reached, and both players settle for their alternatives).

In this paper, we consider a natural {\em network} extension of this
game that was recently introduced by Kleinberg and
Tardos~\cite{KT08}. Here, the set of players corresponds to the
vertices of an undirected graph $G=(V,E)$; each edge $ij \in E$
represents a potential deal between players $i$ and $j$ of unit value.
In Kleinberg and Tardos' model, players are restricted to bargain with
at most one of their neighbours. Outcomes of the {\em network
  bargaining} game (\nwb) are therefore given by a matching $M \subseteq E$,
and an {\em allocation} $x \in \IR^V_+$ such that $x_i+x_j=1$ for all
$ij \in M$, and $x_i=0$ if $i$ is $M$-exposed; i.e., if it is not
incident to an edge of $M$.

Unlike in the non-network bargaining game, the alternative $\alpha_i$
of player is not a given parameter but rather implicitly determined by
the network neighbourhood of $i$. Specifically, in an outcome $(M,x)$,
player $i$'s alternative is defined as
\begin{equation}\label{eq:alpha}
  \alpha_i = \max \{ 1-x_j \,:\, ij \in \delta(i)\setminus M\},
\end{equation}
where $\delta(i)$ is the set of edges incident to $i$. Intuitively,
a neighbour $j$ of $i$ receives $x_j$ in her current deal, and $i$ may
coerce her into a joint deal, yielding $i$ a payoff of
$1-x_j$.

An outcome $(M,x)$ of \nwb\ is called {\em
  stable} if $x_i + x_j \geq 1$ for all edges $ij \in E$, and it is
{\em balanced} if in addition, the value of the edges in $M$ is split
according to Nash's bargaining solution; i.e., for a matching edge $ij$, 
$x_i-\alpha_i=x_j-\alpha_j$. 

Kleinberg and Tardos gave an efficient algorithm to compute balanced
outcomes in a graph (if these exist). Moreover, the authors
characterize the class of graphs that admit such outcomes. In the
following main theorem of \cite{KT08}, a vertex $i \in V$ is called
{\em inessential} if there is a maximum matching in $G$ that exposes
$i$. 

\begin{theorem}[\cite{KT08}]\label{thm:kt}
  An instance of \nwb\ has a balanced outcome
  iff it has a stable one. Moreover, it has a stable outcome iff no
  two inessential vertices are connected by an edge. 
\end{theorem}

The theory of {\em cooperative games} offers another useful angle for
\nwb.  In a cooperative game (with transferable utility) we are given
a player set $N$, and a {\em valuation function} $v:2^N \rightarrow
\IR_+$; $v(S)$ can be thought of as the value that the players in $S$
can jointly create. The {\em matching game}~\cite{DIN99,SS71} is a
specific cooperative game that will be of interest for us. Here, the
set of players is the set of vertices $V$ of a given undirected
graph. The matching game has valuation function $\nu$ where $\nu(S)$
is the size of a maximum matching in the graph $G[S]$ induced by the
vertices in $S$.

One goal in a cooperative game is to allocate the value $v(N)$ of the
so called {\em grand coalition} fairly among the players. The {\em
  core} is in some sense the gold-standard among the solution concepts
that prescribe such a fair allocation: a vector $x \in \IR^N_+$ is in
the core if (a) $x(N)=v(N)$, and (b) $x(S) \geq v(S)$ for all $S
\subseteq N$, where we use $x(S)$ as a short-hand for 
$\sum_{i \in S} x_i$. In the special case of the matching game, this is
seen to be equivalent to the following:
\begin{equation}\label{eq:matcore}
\core(G) = \{ x \in \IR^V_+ \,:\, x(V)=\nu(V) \mbox{ and } x_u +
x_v \geq 1, \, \forall uv \in E\}. 
\end{equation}
Thus, the core of the matching game consists precisely of the set of
stable outcomes of the corresponding \nwb\ game. This was recently
also observed by Bateni et al.~\cite{BH+10} who remarked that the set of
balanced outcomes of an instance of \nwb\ corresponds to the elements
in the intersection of core and {\em prekernel} 
\myifthen{
  (the precise
  definition is not important at this point, and thus deferred to
  Section \ref{sec:faigle})}{
  (e.g., see \cite{CE+11,PS03} for a definition)},of the associated matching game instance.

\ignore{
$$ \prekernel(N,v) = \{ x \in \IR^N_+ \,:\, s_{ij}(x) = s_{ji}(x),
\forall i,j \in N\}, $$
where $s_{ij}(x)$ is the {\em power} of player $i$ over player
$j$. Formally, 
$$ s_{ij} = \max \{v(S)-x(S) \,:\, i \in S, j \not\in S\} $$
is the maximum {\em excess} that }

\subsection{Dealing with unstable instances}

Using the language of cooperative game theory and the work of Bateni
et al.~\cite{BH+10}, we can rephrase the main results of \cite{KT08}
as follows:  \emph{
Given an instance of \nwb, if the core of the underlying
matching game is non-empty then there is an efficient algorithm to
compute a point in the intersection of core and prekernel. } 
Such an
algorithm had previously been given by Faigle et al. in \cite{FKK98}.
It is not hard to see that the core of an instance of the matching
game is non-empty if and only if the fractional matching LP for this
instance has an integral optimum solution. We state this LP and its
dual below; we let $\delta(i)$ denote the set of edges incident to
vertex $i$ in the underlying graph, and use $y(\delta(i))$ as a
shorthand for the sum of $y_e$ over all $e \in \delta(i)$. 

\medskip
\begin{minipage}{0.35\linewidth}
  \begin{align} 
    \max \quad & \sum_{e \in E}y_e \label{lp:mat}\tag{P}\\
    \st \quad & y(\delta(i)) \leq 1 \quad \forall i \in V \notag \\
    & y \geq \0 \notag
  \end{align}
\end{minipage}
\hspace*{1ex}
\vline
\hspace*{1ex}
\begin{minipage}{0.45\linewidth}
  \begin{align}
    \min \quad & \sum_{i \in V}x_i \label{lp:matd}\tag{D}\\
    \st \quad & x_i + x_j \geq 1 \quad \forall ij \in E \label{lp:matd:1} \\
    & x \geq \0,\notag
  \end{align}
\end{minipage}
\bigskip

LP \eqref{lp:mat} does of course typically have a fractional optimal
solution, and in this
case the core of the corresponding matching game
instances is empty. Core assignments are
highly desirable for their properties, but may simply not be available 
for many instances. 
For this reason, a number of more forgiving alternative
solution concepts like {\em bargaining sets, kernel, nucleolus},
etc. have been proposed in the cooperative game theory literature
(e.g., see \cite{CE+11,PS03}).

This paper addresses network bargaining instances that are {\em
  unstable}; i.e., for which the associated matching game has an empty
core. From the above discussion, we know that there is no solution $x$
to \eqref{lp:matd} that also satisfies $\1^Tx \leq \nu(V)$.  We
therefore propose to find an allocation $x$ of $\nu(V)$ that violates
the stability condition in the smallest number of places. Formally,
we call a set $B$ of edges a {\em blocking set} if there is $x \in
\IR_+^V$ such that $\1^Tx \leq \nu(V)$, and $x_i + x_j \geq 1$ for all
$ij \in E\setminus B$.

Blocking sets were previously discussed by Bir\'o et
al.~\cite{BKP10}. The authors showed that finding a smallest such set is
NP-hard (via a reduction from {\em maximum independent set}). In this
paper, we complement this result by showing that approximate blocking
sets can be computed in {\em sparse} graphs. A graph $G=(V,E)$ is {\em
  $\omega$-sparse} for some $\omega \geq 1$ if for all $S\subseteq V$,
the number of edges in the induced graph $G[S]$ is bounded by
$\omega\,|S|$. For example, if $G$
is planar, then we may choose $\omega=3$ by Euler's formula.

%(e.g., see \cite{west}) tells us that
%we may choose $\omega=3$.

\begin{theorem}\label{thm:main}
  Given an $\omega$-sparse graph $G=(V,E)$, there is an efficient algorithm for
  computing blocking sets of size at most $8\omega+2$ times the optimum.
\end{theorem}

The main idea in our algorithm is a natural one: formulate the
blocking set problem as a linear program, and extract a blocking set
from one of its optimal fractional solutions via an application of the
powerful technique of  {\em iterative rounding} (e.g.,
see \cite{LRS11}). We first show that the proposed LP has an unbounded
integrality gap in general graphs, and is therefore not useful for the
design of approximation algorithms for such instances. We turn to
the class of sparse graphs,
and observe that, even here, extreme
points of the LP can be highly fractional, ruling out the direct use
of standard techniques. We carefully characterize problem
extreme-points, and develop a direct rounding method for them. Our
approach exploits problem-specific structure as well as the
sparsity of the underlying graph. 

Given a blocking set $B$, let $E'=E\setminus B$ be the non-blocking
set edges, and let $G'=(V,E')$ be the induced graph. Notice that the
matching game induced by $G'$ may {\em still} have an empty core, and
that the maximum matching in $G'$ may even be smaller than that in
$G$.  We are however able to show that we can find a balanced
allocation of $\nu(V)$ as follows: let $M'$ be a maximum matching
in $G'$, and define the alternative of player $i$ as
$$ \alpha'_i = \max \{1-x_i \,:\, ij \in \delta_{G'}(i) \setminus M'\}, $$
for all $i \in V$. Call an assignment $x$ is balanced if it
satisfies the stability condition \eqref{lp:matd:1} for all edges $ij
\in M'$, and 
$$ x_i - \alpha'_i = x_j - \alpha'_j, $$
for all $ij \in M'$. A straight-forward application of an algorithm of
Faigle et al.~\cite{FKK98} yields a polynomial-time method to compute
such an allocation. 
\myifthen{We describe this algorithm in Section \ref{sec:faigle} for
  completeness.}{Details are omitted from this extended abstract.}

\section{Finding small blocking sets in sparse graphs}
%%%%%%%%%%%%%%%%%%%%%%%%%%%%%

We attack the problem of finding a small blocking set via iterative
linear programming rounding. In order to do this, it is convenient to
introduce a slight generalization of the blocking set problem. 
In an instance of the {\em generalized blocking set} problem
(GBS), we are given a graph $G=(V,E)$, a partition $E_1 \cup E_2$ of
$E$, and a parameter $\nu \geq 0$. The goal is to find a blocking set
$B \subseteq E_1$, and an allocation $x \in \IR^V_+$ such that $\1^Tx
\leq \nu$ and $x_u + x_v \geq 1$ for all $uv \in E\setminus B$, where
$\1$ is a vector of $1$s of appropriate dimension. The problem is
readily formulated as an integer program. We give its relaxation below
on the left.

\medskip
\hspace*{-6.5ex}
\begin{minipage}{0.45\linewidth}
  \begin{align} 
    \min \quad & \1^Tz  \label{lp:bs}\tag{P$_B$}\\
    \st \quad & x_u + x_v + z_{uv} \geq 1 \notag \\
            & \quad\quad\quad\quad \forall uv \in E_1 
    \label{lp:e1} \\
    & x_u + x_v \geq 1 \notag \\
    & \quad\quad\quad\quad \forall uv \in E_2 \label{lp:e2} \\
    & \1^Tx \leq \nu \label{lp:nu} \\
    & x,z \geq \0 \notag
  \end{align}
\end{minipage}
\hspace*{.5ex}
\vline
\begin{minipage}{0.55\linewidth}
  \begin{align}
    \max \quad & \1^Ta + \1^Tb - \gamma\,\nu \label{lp:bsd}\tag{D$_B$}\\
    \st \quad & a(\delta_{E_1}(u)) + \notag \\
          & \quad b(\delta_{E_2}(u)) \leq \gamma
       && \forall u \in V \label{lp:bsd:1} \\
       & a \leq \1 \notag \\
       & a,b \geq \0\notag
  \end{align}
\end{minipage}
\bigskip

The LP on the right is the dual of \eqref{lp:bs}. It has a variable
$a_e$ for all $e \in E_1$, a variable $b_e$ for all $e \in E_2$, and
variable $\gamma$ corresponds to the primal constraint limiting
$\1^Tx$.  
%%%%%%%%%%%%%%%%%%%%%%%%%%%%%%%%%%%%%%%%%%%%%%%%%%%%%%
\myifthen{
We first show that LP \eqref{lp:bs} is not useful for
finding approximate blocking sets in general graphs. To see this, we
will consider bipartite instances of the following form:

\begin{figure}[h]
\begin{center}
\includegraphics[scale=.8]{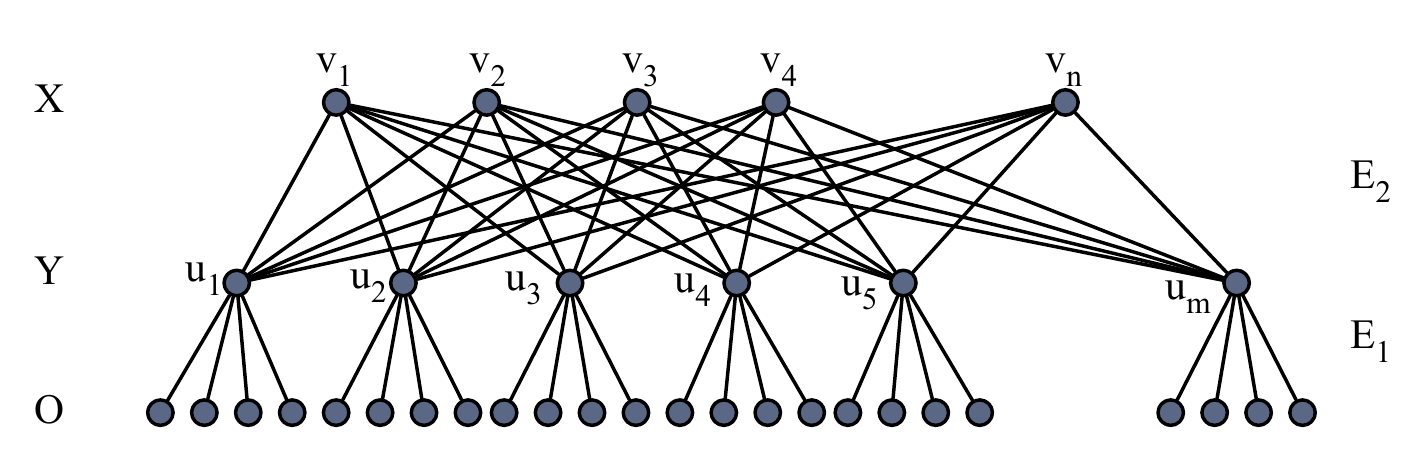}
\end{center}
\end{figure}

Our graphs will have a set $V=X \cup Y \cup O$, where $X$ has $n$, $Y$
has $m$, and $O$ has $4m$ vertices, respectively. The edges connecting
sets $X$ and $Y$ belong to $E_2$ and induce the complete bipartite
graph $K_{n,m}$. The edges between $Y$ and $O$ belong to the set
$E_1$. In our gap instance, we assume $m=2n$, and let $\nu=2n-1$.  In
order to obtain a fractional solution to \eqref{lp:bs}, we let $ x_v =
1-\alpha$ if $v \in X$, $x_v=\alpha$ if $v \in Y$, and $x_v=0$
otherwise, for any $v \in V$, where $\alpha=(n-1)/n$.  We also let
$z_{uv}=1/n$ for all $u\in Y$, and $v\in O$. The solution is clearly
feasible as $\1^Tx=2n-1$, all $E_2$ edges are covered, and $z$ is
sufficiently large to ensure that $E_1$ constraints are satisfied. The
value of this solution is $8n (1-\alpha)=8$. We now show that any
solution $(x,z)$ to \eqref{lp:bs} for which $z$ is binary must have
value at least $4n$. We do this by induction on $n$.

For $n=1$, let $v_1$ be the only $X$ vertex.  Since we have $\nu=1$,
the only feasible solution to \eqref{lp:bs} is given by $x_{v_1}=1$,
and $x_u=0$ otherwise. This forces $z_{uv}=1$ for all $u \in Y, v \in
O$, and thus the value of this solution is $8$. 
Now consider the case where $n>1$. Let $u_1 \in Y$ be a vertex with
$x_{u_1}=1$. If no such vertex exists, then we are done as
then the blocking set has size at least $2n \cdot 4 = 8n$. Let $u_2
\in Y$ be a vertex with $x_{u_2}<1$; such a vertex must also
exist as the total $x$-value on the vertices is bounded by $2n-1$. Now
consider any vertex $v_1 \in X$, and note that
$$ x_{v_1} + x_{u_1} + x_{u_2} \geq 2, $$
by the feasibility of $(x,z)$. Consider the
graph $G'$ induced by $X'=X\setminus \{v_1\}$, $Y'=Y\setminus
\{u_1,u_2\}$, and the neighbours $O' \subseteq O$ of vertices in
$Y'$, and let $\nu'=2(n-1)-1$. By induction, we know that a feasible integral
solution $(x',z')$ for this instance must have value at least
$4(n-1)$. As
$$ x(X' \cup Y' \cup O') \leq 2n - 1 - (x_{v_1} + x_{u_1} + x_{u_2})
\leq 2(n-1)-1, $$
there are at least $4(n-1)$ edges in $G'$ that are not covered by
$x$. Finally, $x$ does not cover the $4$ edges incident to $u_2$, and
hence $(x,y)$ has value at least $4n$. 
}{
We can show the LP is weak and hence not useful for approximating the
generalized blocking set problem in general graphs (for details, see \cite{KLS12}).
}
%%%%%%%%%%%%%%%%%%%%%%%%%%%%%%%%%%%%%%%%%%%%%%%%%%%%%%

\begin{lemma}
  The integrality gap of \eqref{lp:bs} is $\Omega(n)$, where $n$ is
  the number of vertices in the given instance of the blocking set problem. 
\end{lemma}

Given this negative result, we will focus on sparse instances
$(G,\nu)$ and prove Theorem \ref{thm:main}. We first characterize
the extreme points of \eqref{lp:bs}. 

\subsection{Extreme points of \eqref{lp:bs}}
%%%%%%%%%%%%%%%%%%%%%%%

In the following, we assume that the underlying graph $G$
is bipartite; this assumption will greatly simplify our
presentation, and will turn out to be w.l.o.g.
Let $(x,z)$ be a feasible solution of LP \eqref{lp:bs}, and let
$A^=(x,z)^T=b^=$ be the set of tight constraints of the LP. It is well
known (e.g., see \cite{Sc86} and also \cite{LRS11}) that $(x,z)$ is an
extreme point of the feasible region if $A^=$ has full column-rank. In
particular, $(x,z)$ is uniquely determined by any full-rank sub-system
$A'(x,z)^T=b'$ of $A^=(x,z)^T=b^=$. If constraint \eqref{lp:nu} is not
part of this system of equations, then
$$ A' = [A'', I], $$
where $A''$ is a submatrix of the edge-vertex incidence matrix of a
bipartite graph, and $I$ is an identity matrix of appropriate
dimension. Such matrices $A'$ are well-known to be {\em totally
  unimodular} (e.g., see \cite{Sc86}), and $(x,z)$ is therefore integral in this
case.  From now on, we therefore assume that constraint \eqref{lp:nu} is tight,
and that $(x,z)$ is the unique solution to 
\begin{equation}\label{eq:extpt}
\left[ \begin{array}{cc} A'' & I \\ \1^T & \0^T \end{array}\right]
\begin{pmatrix} \bar{x} \\ \bar{z} \end{pmatrix} = 
\begin{pmatrix} \1 \\ \nu \end{pmatrix}, 
\end{equation}
where $A''$ is a submatrix of the edge, vertex incidence matrix of
bipartite graph $G$, $I$ is an identity matrix, and $\1^T$ and
$\0^T$ are row vectors of $1$'s and $0$'s, respectively. We obtain the
following useful lemma.

\begin{lemma}\label{lem:alpha}
  Let $(x,z)$ be a non-integral extreme point solution to
  \eqref{lp:bs} satisfying \eqref{eq:extpt}. Then there is an $\alpha
  \in (0,1)$ such that $x_u, z_{uv} \in \{0,\alpha,1-\alpha,1\}$ for
  all $u \in V$, and $uv \in E_1$.
\end{lemma}
\begin{proof}
  Standard linear algebra implies that the solution space to the  
  the system $[A'' \, I] (\bar{x},\bar{z})^T$ is a line; i.e., it has
  dimension $1$. Hence, there are two extreme points $(x^1,z^1)$ and
  $(x^2,z^2)$ of the integral polyhedron defined by constraints
  \eqref{lp:e1}, \eqref{lp:e2}, and the non-negativity constraints,
  and some $\alpha \in [0,1]$ such that 
  $$ \begin{pmatrix} x \\ z \end{pmatrix} =
       \alpha \, \begin{pmatrix} x^1 \\ z^1 \end{pmatrix} + 
       (1-\alpha) \, \begin{pmatrix} x^2 \\ z^2 \end{pmatrix}.  
  $$
  In fact, $\alpha$ must be in $(0,1)$ as $(x,z)$ is assumed to be
  fractional. This implies the lemma.\qed
\end{proof}

We call an extreme point {\em good} if there is a vertex $u$ with
$x_u=1$, or an edge $uv \in E_1$ with $z_{uv} \in \{0\} \cup
[1/3,1]$. Let us call an extreme point {\em bad} otherwise.  We will
now characterize the structure of a bad extreme point $(x,z)$. Let
$G=(V,E_1 \cup E_2)$ be the bipartite graph for a given GBS
instance. Let $\T_1\subseteq E_1$ and $\T_2\subseteq E_2$ be $E_1$ and
$E_2$ edges corresponding to tight inequalities of \eqref{lp:bs} that
are part of the defining system \eqref{eq:extpt} for $(x,z)$.  Let
$\alpha$ be as in Lemma \ref{lem:alpha}. Since $(x,z)$ is bad, it must
be that either $\alpha$ or $1-\alpha$ is larger than $2/3$; w.l.o.g.,
assume that $\alpha > 2/3$. We define the following useful sets:
\begin{eqnarray*}
  X & = & \{u \in V \,:\, x_u = 1-\alpha\} \\
  Y & = & \{u \in V \,:\, x_u = \alpha \} \\
  O & = & \{u \in V \,:\, x_u = 0\}. 
\end{eqnarray*}

\begin{lemma}\label{lem:bad}
  Let $(x,z)$ be a bad extreme point. Using the notation defined
  above, we have
  \begin{enumerate}[(a)]
    \item $z_{uv}=(1-\alpha)$ for all $uv \in E_1$, 
    \item $O \cup X$ is an independent set in $G$
    \item Each $\T_1$ edge is incident to exactly one $O$ and one $Y$ vertex,
      and the edges of $\T_2$ form a tree spanning $X \cup Y$.
      Each edge in $E$ is incident to exactly one $Y$ vertex.
 \end{enumerate}
\end{lemma}
\begin{proof}
  We know from Lemma \ref{lem:alpha} that $z_{uv} \in \{0,1-\alpha,
  \alpha, 1\}$ for all $uv \in E_1$; (a) follows now directly from the
  fact that $(x,z)$ is bad. 

  No two vertices $u,v \in O$ can be connected by an edge, as such an
  edge $uv$ must then have $z_{uv}=1$. Similarly, no two vertices $u,v
  \in X$ can be connected by an edge as otherwise $z_{uv} \geq
  1-2(1-\alpha) > 1/3$. Finally, for an edge $uv$ between $O$ and $X$,
  we would have to have $z_{uv} \geq 1-(1-\alpha) > 2/3$, which once
  again can not be the case. This shows (b). 
  
  To see (c), consider first an edge $uv$ in $\T_1$; we must have
  $x_u+x_v=\alpha$, and this is only possible if $uv$ is incident to
  one $O$ and one $Y$ vertex. Similarly, $x_u+x_v=1$ for all $uv \in
  \T_2$, and therefore one of $u$ and $v$ must be in $X$, and one must
  be in $Y$. It remains to show that the edges in $\T_2$ induce a
  tree. Let us first show acyclicity: suppose for the sake of
  contradiction that $u_1v_1, \ldots, u_pv_p \in \T_2$ form a cycle
  (i.e., $u_1=v_p$). Then since $G$ is bipartite, this cycle contains
  an even number of edges. Let $\chi_1, \ldots, \chi_p$ be the
  $0,1$-coefficient vector of the left-hand sides of the constraints
  belonging to these edges. We see that
  $$ \sum_{i=1}^p (-1)^i \chi_i = \0, $$
  contradicting the fact that the system in \eqref{eq:extpt} has full
  (row) rank. Note that the size of the support of $(x,z)$ is
  \begin{equation}\label{eq:supsz}
    |\T_1|+|X|+|Y|
  \end{equation}
  by definition. On the other hand, the rank of the system in
  \eqref{eq:extpt} is 
  $$ |\T_1| + |\T_2| + 1 \leq |\T_1| + (|X|+|Y|-k) + 1,  $$
  where $k$ is the number of components formed by the edges in $\T_2$.
  The rank of \eqref{eq:extpt} must be at least the size of the
  support, and this is only the case when $k=1$; i.e., when $\T_2$
  forms a tree spanning $X \cup Y$. 
  Since $G$ is bipartite, $X$ must
  be fully contained in one side of the bipartition of $V$, and $Y$
  must be fully contained in the other. Since $Y$ is a vertex cover in
  $G$ by (b), every edge in $E$ must have exactly one endpoint in
  $Y$. 
  \qed\end{proof}

\myifthen{
\begin{figure}[ht]
  \begin{center}
    \includegraphics[scale=0.9]{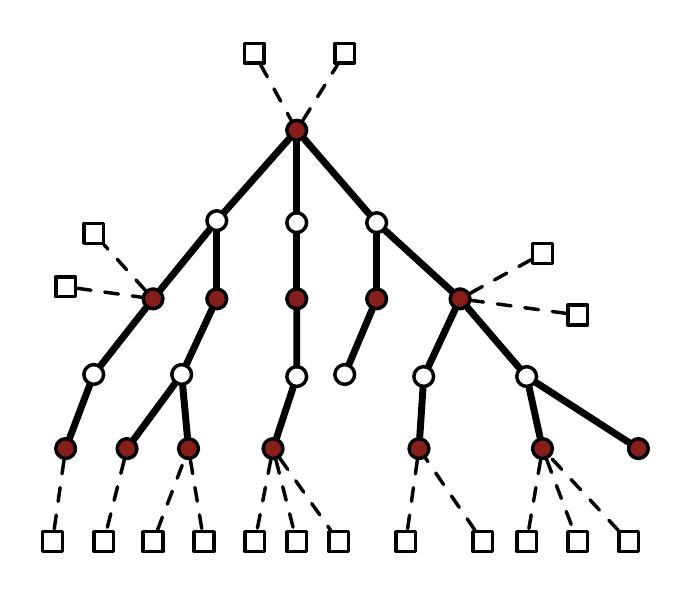}
  \end{center}
  \caption{\label{fig:badext}The figure shows the structure of a bad
    extreme point. White nodes correspond to $X$ vertices, dark ones
    to vertices in $Y$, and the squares are $O$-vertices. Edges in
    $\T_1$ are displayed as dashed lines, and those in $\T_2$ as
    solid, thick ones.}
\end{figure}}

\subsection{Blocking sets in sparse graphs via iterative rounding}

In this section we propose an iterative rounding (IR) type algorithm
to compute a blocking set in a given sparse graph $G=(V,E)$. Recall
that this means that there is a fixed parameter $\omega > 0$ such that
the graph induced by any set $S$ of vertices has at most $\omega |S|$
edges. Recall that we also initially assume that the underlying graph
$G$ is bipartite.

The algorithm we propose follows the standard IR paradigm (e.g., see
\cite{LRS11}) in many ways: given some instance of the blocking set
problem, we first solve LP \eqref{lp:bs} and obtain an extreme point
solution $(x,z)$. We now generate a {\em smaller} sub-instance of GBS
such that (a) the {\em projection} of $(x,z)$ onto the sub-instance is
feasible, and (b) any integral solution to the sub-instance can
cheaply be extended to a solution of the original GBS instance. In
particular, the reader will see the standard steps familiar from other
IR algorithms: if there is an edge $uv \in E_1$ with $z_{uv}=0$ then
we may simply drop the edge, if $z_{uv} \geq 1/3$ then we include the
edge into the blocking set, and if $x_u=1$ for some vertex, then we
may install one unit of $x$-value at $u$ permanently and delete $u$
and all incident edges.

The problem is that the feasible region of \eqref{lp:bs} has bad extreme points, even if the
underlying graph is sparse and bipartite. We will exploit the
structural properties documented in Lemma \ref{lem:bad} and show that
a small number of edges can be added to our blocking set even in this
case. Crucially, these edges will have to come from both $E_1$ and $E_2$.

In an iteration of the algorithm, we are given a sub-instance of
GBS. We first solve \eqref{lp:bs} for this instance, and obtain an
optimal basic solution $(x,z)$. Inductively we maintain the following: 
The algorithm computes a set $\hat{B} \subseteq E$ of edges,
and vector $\hat{x} \in \IR^V$ such that 
\begin{enumerate}[{[{I}1]}]
  \item $\hat{x}_u + \hat{x}_v \geq 1$ for all $uv \in E\setminus \hat{B}$,
  \item $\1^T\hat{x} \leq \nu$, and 
  \item $|\hat{B}| \leq (2\omega+1) \cdot \1^Tz$,
\end{enumerate}
where $\omega$ is the sparsity parameter introduced above. 
Let us first assume that the extreme point solution $(x,z)$ is
good. In this case we
proceed according to one of the following cases:
\begin{description}
\item[{\bf Case 1.}] ($\exists u \in V$ with $x_u=1$) In this case, all edges
incident to $u$ are covered. We obtain a subinstance of
GBS by removing $u$ and all incident edges from $G$, and by reducing
$\nu$ by $1$.
\item[{\bf Case 2.}] ($\exists uv \in E$ with $z_{uv}=0$) In this case,
obtain a new instance of GBS by removing $uv$ from $E_1$, and adding
it to $E_2$. 
\item[{\bf Case 3.}] ($\exists uv \in E_1$ with $z_{uv}\geq 1/3$) In this case
add $uv$ to the approximate blocking set $B$, and remove $uv$ from
$E_1$. 
\end{description}
In each of these three cases, we inductively solve the generated
sub-instance of GBS. If this subinstance is the empty graph, then we
can clearly return the empty set.

Let us now consider the case where $(x,z)$ is a bad extreme point.
This case will constitute a leaf of the recursion tree, and we will
show that we can directly find a small blocking set.  In the
following lemma, we define the sets $X, Y, O \subseteq V$ as in Lemma
\ref{lem:bad}. 
\myifthen{}{Its proof is deferred to \cite{KLS12}.}

\begin{lemma}\label{lem:nubd}
  Let $(x,z)$ be a bad extreme point, and let $\nu$ be the current
  bound on $\1^Tx$. Then $(|X|+|Y|)/2 < \nu < |Y|$.
\end{lemma}

\myifthen{
\begin{proof}
Lemma \ref{lem:bad} (b) shows that $Y$ is a vertex cover in the
  current graph $G$. Hence, if $\nu \geq |Y|$ then we could simply let
  $\hat{x}_u=1$ for all $u \in Y$ and choose $B$ to be the empty
  blocking set. This also implies that $|Y| > |X|$ as otherwise
  $$ \1^Tx = (1-\alpha)|X|+\alpha |Y| \geq |Y| > \nu, $$
  contradicting feasibility. 
  To see the lower-bound, recall that \eqref{lp:nu} is tight by
  assumption, and thus
  $$ \nu = (1-\alpha)|X| + \alpha |Y| > \frac{2|Y| + |X|}{3} >
  \frac{|X|+|Y|}{2}, $$
  where the first inequality uses the fact that $\alpha > 2/3$ from
  Lemma \ref{lem:bad}, and the last inequality follows as $|Y| >
  |X|$.
\qed\end{proof}}{}

We can use this bound on $\nu$ to prove that we can find small
blocking sets given a bad extreme point for \eqref{lp:bs}.

\begin{lemma}\label{lem:badextbd}
  Given a bad extreme point $(x,z)$ to \eqref{lp:bs}, we can find a
  blocking set $\hat{B}\subseteq E$, and corresponding $\hat{x}$ such that
  $\1^T\hat{x} \leq \nu$, and $|\hat{B}| \leq (2\omega+1) \cdot \1^Tz$. 
\end{lemma}
\begin{proof}
     We will construct a blocking set $\hat{B}$ as follows: let $\hat{x}_u=1$ for a
  carefully chosen set $\hat{Y}$ of $\nu$ vertices from the set $Y$, and let
  $\hat{x}_u=0$ for all other vertices in $V$. Recall once more from
  Lemma \ref{lem:bad} (b) that $Y$ is a vertex cover in $G$, and hence
  it suffices to choose
  \begin{equation}\label{eq:Bdef}
    \hat{B} = \bigcup_{u \in Y\setminus\hat{Y}} \delta(u) = \bigcup_{u \in
      Y\setminus \hat{Y}} \bigl( \delta_{E_1}(u) + \delta_{E_2}(u) \bigr)
  \end{equation}
  as our blocking set, where $\delta_{E_i}(u)$ denotes the set of
  $E_i$ edges incident to vertex $u$. Let $(a,b,\gamma)$ be the
  optimal dual solution of \eqref{lp:bsd} corresponding to extreme
  point $(x,z)$. Then note that complementary slackness together with
  the fact that $z_{uv} > 0$ for all $uv \in E_1$ implies that
  $a_{uv}=1$ for these edges as well. Thus $\gamma$ is an upper bound
  on the number $E_1$-edges incident to a vertex $u$ by dual
  feasibility. With \eqref{eq:Bdef} we therefore obtain
  \begin{equation} \label{eq:Bdef2}
    |\hat{B}| \leq \sum_{u \in Y\setminus \hat{Y}} \left(\gamma +
      |\delta_{E_2}(u)| \right) \leq (|Y|-\nu)\gamma + 
    \sum_{u \in Y\setminus\hat{Y}} |\delta_{E_2}(u)|. 
  \end{equation}
  Lemma \ref{lem:bad} (c) shows that each $E_2$ edge is incident to
  one $X$, and one $Y$ vertex. As the subgraph induced by $X$ and $Y$ is
  sparse, there therefore must be a
  vertex $u_1 \in Y$ of degree at most $\omega
  (|X|+|Y|)/|Y|$. Removing this vertex from $G$ leaves a sparse graph,
  and we can therefore find a vertex $u_2$ of degree at most $\omega
  (|X|+|Y|-1)/(|Y|-1)$. Repeating this $|Y|-\nu$ times we pick a set
  $u_1, \ldots, u_{|Y|-\nu}$ of vertices such that
  \myifthen{\begin{equation}}{\begin{multline}}\label{eq:Bdef3}
    \sum_{i=1}^{|Y|-\nu} |\delta_{E_2}(u_i)| \leq
  \sum_{i=1}^{|Y|-\nu} \frac{\omega (|X|+|Y|-i)}{|Y|-i} \leq \\
  (|Y|-\nu)
  \cdot \frac{\omega (|X|+|Y|)}{\nu} \leq 2\omega (|Y|-\nu), 
  \myifthen{\end{equation}}{\end{multline}}
  where the last inequality follows from Lemma \ref{lem:nubd}.
  We now let $\hat{Y} = Y\setminus \{u_1, \ldots, u_{|Y|-\nu}\}$, and
  hence let $\hat{x}_{u}=1$ for $u \in \hat{Y}$, and $\hat{x}_u=0$ for
  all other vertices $u \in V$; 
  \eqref{eq:Bdef2} and \eqref{eq:Bdef3} together imply that 
  $$ |\hat{B}| \leq (|Y|-\nu) (\gamma + 2\omega) \leq (2\omega+1)
  \gamma (Y-\nu), $$ where the last inequality follows from the fact
  that $\gamma \geq 1$. Lemma \ref{lem:bad}(c) shows
  that each edge $e \in E$ has exactly one endpoint in $Y$. Applying
  complementary slackness together with the fact that $x_u>0$ for all
  $u \in Y$, we can therefore rewrite the objective function of
  \eqref{lp:bsd} as
  $$ \1^Ta + \1^Tb - \gamma\,\nu = \gamma (|Y|-\nu). $$
 The lemma follows.
\qed\end{proof}

We can now put things together.

\begin{lemma}\label{lem:final}
  Given an instance of GBS, the above procedure terminates with a set
  $\hat{B} \subseteq E$, and $\hat{x} \in \IR^V$ such that $\1^T\hat{x} \leq
  \nu$, and $\hat{x}_u+\hat{x}_v \geq 1$ for all $uv \in E\setminus \hat{B}$.
  The set $\hat{B}$ has size at most $(2\omega+1)\1^Tz$, where $(x,z)$ is an
  optimal solution to \eqref{lp:bs} for the given GBS instance.
\end{lemma}
\begin{proof}
  The proof uses the usual induction on the recursion depth. Let us
  first consider the case where the current instance is a leaf of the
  recursion tree. The lemma follows vacuously if the graph in the given GBS
  instance is empty. Otherwise it follows immediately from Lemma
  \ref{lem:badextbd}. 

  Any internal node of recursion tree corresponds to an instance of
  GBS where $(x,z)$ is a good extreme point.  We claim that, no matter
  which one of the above cases we are in, we have that (a) a suitable
  projection of $(x,z)$ yields a feasible solution for the created GBS
  sub-instance, and (b) we can {\em augment} an approximate blocking
  set for this sub-instance to obtain a {\em good} blocking set for
  the instance given in this iteration. We proceed by looking at the
  three cases discussed above.

  \noindent {\bf Case 1.} Let $(x',z')$ be the natural projection of
  $(x,z)$ onto the GBS sub-instance; i.e., $x'_v$ is set to $x_v$ for
  all vertices in $V-u$, and $z'_{vw}=z_{vw}$ for the remaining edges
  $vw \in E_1\setminus \delta(u)$. This solution is easily verified to
  be feasible. Inductively, we therefore know that we obtain a
  blocking set $\bar{B}$ and corresponding vector $\bar{x}$ such that
  $\bar{B}$ has no more than $(2\omega+1)\, \1^T\bar{z} \leq
  (2\omega+1)\, \1^Tz$ elements, and $\1^T\bar{x} \leq \nu-1$. Thus,
  letting $\hat{x}_v=\bar{x}_v$ for all $v \in V-u$, and $\hat{x}_u=1$
  together with $\hat{B}=\bar{B}$ gives a feasible solution for the
  original GBS instance.
  
  \noindent{\bf Case 2.} The argument for this case is virtually
  identical to that of Case 1, and we omit the details. 

  \noindent{\bf Case 3.} Once again we project the current solution
  $(x,z)$ onto the GBS subinstance; i.e., let $x'=x$, and
  $z'_{qr}=z_{qr}$ for all $qr \in E_1 - uv$. Clearly
  $(x',z')$ is feasible for the GBS subinstance, and
  inductively we therefore obtain a vector $\bar{x}$ and corresponding
  feasible blocking set $\bar{B}$ of
  size at most $(2\omega+1)\cdot \1^Tz'$. Adding $uv$ to $\bar{B}$ yields a
  feasible blocking set $\hat{B}$ for the original instance together with
  $\hat{x}=\bar{x}$. Its size is at most $(2\omega+1)\,\1^Tz' + 1 \leq (2\omega+1)\,
  1^Tz$ as $\omega \geq 1$.
\qed\end{proof}

Suppose now that we are given a non-bipartite, sparse instance of the
blocking set problem:
$G=(V,E)$ is a general sparse graph, and $\nu > 0$ is a parameter. We
create a bipartite graph $H$ in the usual way: for each vertex $u \in V$
create two copies $u_1$ and $u_2$ and add them to $H$. For each edge
$uv \in E$, add two edges $u_1v_2$ and $u_2v_1$ to $H$. The new
blocking set instance is given by $(H,\nu')$ where $\nu'=2\nu$. 

Given a feasible solution $(x,z)$ to \eqref{lp:bs} for the instance
$(G,\nu)$ , we let $x'_{u_i}=x_u$ for all
$u \in V$ and $i \in \{1,2\}$, and $z'_{u_iv_j}=z_{uv}$ for all edges
$u_iv_j$. For any edge $u_iv_j$ in $H$, we now have
$$ x'_{u_i} + x'_{v_j} + z_{u_iv_j} = x_u + x_v + z_{uv} \geq 1, $$
and $\1^Tx' \leq 2\1^Tx \leq 2\nu$. Thus, 
$(x',z')$ is feasible to \eqref{lp:bs} for instance $(H,\nu')$, and
its value is at most twice that of $\1^Tz$. 
Let $\hat{x}, \hat{B}$ be a feasible solution to
the instance on graph $H$. Then let 
$$ B = \{ uv \in E \,:\, u_1v_2 \mbox{ or } u_2v_1 \mbox{ are in }
\hat{B}\}, $$
and note that $B$ has size at most that of $\hat{B}$. Also let
$x_u=(\hat{x}_{u_1}+\hat{x}_{u_2})/2$ for all $u \in V$. Clearly,
$\1^Tx \leq \nu$, and for any edge $uv \in E$, we have
$$ x_u + x_v \geq \frac{\hat{x}_{u_1} + \hat{x}_{u_2} + \hat{x}_{v_1}
  + \hat{x}_{v_2}}{2}, $$
and the right-hand side is at least $1$ if none of the two edges $u_1v_2$,
$u_2v_1$ is in $\hat{B}$. This shows feasibility of the pair $x,B$. 
In order to prove Theorem \ref{thm:main} it now remains to show that
graph $H$ is sparse. Pick any set $S$ of vertices in $H$, and let 
$$ S'=\{v \in V \,:\, \mbox{at least one of } v_1 \mbox{ and } v_2
\mbox{ are in } S\}. $$ Then $|S'| \leq |S|$, and the number of edges
of $H[S]$ is at most twice the number of edges in $G[S']$, and hence
bounded by $2\omega\, |S|$; we let $\omega'=2\omega$ be the sparsity
parameter of $H$. 
Let $(x,z)$ and $(x',z')$ be optimal basic solutions to \eqref{lp:bs}
for instances $(G,\nu)$, and $(H,\nu')$, respectively. The blocking
set $B$ for $G$ has size no more than
$$ (2\omega' + 1) \1^Tz' \leq 2(4\omega + 1) \1^Tz. $$

Thus, we have proven Theorem~\ref{thm:main}.

%\begin{theorem}\label{thm:main}
%  Given an $\omega$-sparse graph $G=(V,E)$, there is an efficient algorithm for
%  computing blocking sets of size at most $8\omega+2$ times the optimum.
%\end{theorem}

\myifthen{
\section{From blocking set to balanced allocation}\label{sec:faigle}
%%%%%%%%%%%%%%%%%%%%%%%%%%%%%%%%%%%%

Let $G=(V,E)$ be an instance of the matching game, $B \subseteq E$ a
blocking set, and $x \in \IR^V_+$ s.t.
\begin{align}
  \1^Tx & \leq \nu(V) \tag{S1} \label{stab:1} \\
  x_i + x_j&\geq 1 \quad \forall ij \in E\setminus B.
  \tag{S2} \label{stab:2}
\end{align}
In this section we will show that, for any maximum matching $M'$ in
$G'=G[E\setminus B]$, we can efficiently find a vector $\bar{x}$ that
satisfies \eqref{stab:1}, \eqref{stab:2}, as well as the balancedness
condition
\begin{equation}\tag{S3} \label{stab:3}
  \bar{x}_i - \alpha_i = \bar{x}_j - \alpha_j \quad \forall ij \in M',
\end{equation}
thus resulting in a stable outcome.

Following the work of Bateni et al. \cite{BH+10} this task reduces to
finding a point in the intersection of core and prekernel of the
matching game. We describe an elegant algorithm due to Faigle et
al.~\cite{FKK98} (in the matching game special case); building on a
local-search algorithm due to Maschler, and Stearns~\cite{St68}, the
authors presented an algorithm that efficiently computes prekernel
elements of general TU games (under mild conditions).

\subsection{From prekernel to balancedness}

We define the prekernel of the matching game first. For a pair of
vertices $i,j \in V$, define the {\em surplus} of $i$ over $j$ as
$$ s_{ij}(x) = \max \{ 1 -x_k -x_i \,:\, ik \in E', \, k\neq j\}, $$
where $E'=E\setminus B$ is the set of non-blocking set edges in
$G$\footnote{We abuse notation slightly in this definition of
  $s_{ij}$, and let $s_{ij}(x)=-x_i$ if $i$ has no neighbour other
  than $j$.}.  We will omit the argument $x$ if it is clear from the
context.  For convenience, we will let $e_{ij}$ and $e_{ji}$ be fixed
edges incident to vertices $i$ and $j$ that define the respective
surpluses; i.e. $s_{ij}(x)=1-x(e_{ij})$ and $s_{ji}(x)=1-x(e_{ji})$,
where $x(uv)$ is the {\em value} of edge $uv$ and is a short-hand for
the sum $x_u+x_v$ of the values of the incident vertices.  The {\em prekernel} of the
matching game is the set of all non-negative vectors $x$ with $\1^Tx=\nu(V)$,
and $s_{ij}(x)=s_{ji}(x)$ for all $i,j \in V$.

  Consider some matching $M' \in G[E']$ and let $\alpha_u$ be the
  alternative of $u$ in graph $G'$ with respect to this matching (see
  \eqref{eq:alpha}). Then one sees that $\bar{x}_i - \alpha_i =
  \bar{x}_j-\alpha_j$ if and only if $s_{ij}(\bar{x}) =
  s_{ji}(\bar{x})$. Hence, we will now describe an algorithm that
  finds a vector $\bar{x}$ satisfying \eqref{stab:1}, \eqref{stab:2},
  and the following reformulation of \eqref{stab:3}:
\begin{equation}\label{stab:4}\tag{S3'}
  s_{ij} = s_{ji} \quad \forall ij \in E'. 
\end{equation}
From now on, we will assume that $s_{ij} \geq s_{ji}$ for all $ij \in
E'$. Suppose that $x$ satisfies \eqref{stab:1} and \eqref{stab:2}, but
not \eqref{stab:4}, and let
$$ s_{i_1j_1} \geq s_{i_2j_2} \geq \ldots \geq s_{i_mj_m} $$
be a list of the surpluses of edges in $E'$. Let $1 \leq p \leq m$ be
smallest such that $s_{i_pj_p} > s_{j_pi_p}$, and define
$s(x)=s_{i_pj_p}$ as the largest surplus of any violated pair. 
Let $S(x)$ and $I(x)$ be the sets of all pairs whose surplus is larger
than or equal to $s(x)$, respectively; i.e., 
\begin{eqnarray*}
  S(x) & = & \{ ij \,:\, s_{ij}(x) > s(x) \} \\
  I(x) & = & \{ ij \,:\, s_{ij}(x) = s(x) \}.
\end{eqnarray*}
Let $\mu > 0$ such that
$s_{i_pj_p} - 2\mu = s_{j_pi_p}$ for some $1 \leq p \leq m$. 
We then consider the following natural local shift proposed by Maschler:
$$ x'_u = \begin{cases} x_u + \mu & u = i_p \\ x_u - \mu & u = j_p \\
  x_u & \mbox{otherwise,} \end{cases} $$ for all $u \in V$. 
We collect a few useful observations in the following lemma.

\begin{lemma}[\cite{FKK98}] \label{lem:maschler}
  \begin{enumerate}[(i)]
  \item $s_{i_pj_p}(x')=s_{j_pi_p}(x')=s(x)-\mu$
  \item $s(x') \leq s(x)$, and if 
    $s(x')=s(x)$ then $|I(x')| < |I(x)|$
  \item For $ij \in S(x)$, none of $e_{ij}$ and $e_{ji}$ change value, 
    and if some edge $e \neq ij$ incident to $i$ or $j$ decreases in
    value, then its new value is larger than that of $e_{ij}$ and
    $e_{ji}$. Hence, $s_{ij}(x')=s_{ij}(x)$ for all $ij \in S(x)$.
  \item $x'$ satisfies \eqref{stab:1} and \eqref{stab:2} if $x$ does.
  \end{enumerate}
\end{lemma}
\begin{proof}
  Note that $x'(e)=x(e) + \mu$ for all $e \in \delta(i_p) \setminus
  \{i_pj_p\}$, and hence $s_{i_pj_p}(x')=s_{i_pj_p}(x) -
  \mu$. Similarly, all edges $e \neq i_pj_p$ incident to $j_p$ have
  $x'(e)=x(e)-\mu$, and hence $s_{j_pi_p}(x')=s_{j_pi_p}(x)+\mu$. This
  implies (i). 

  \parpic(5cm,3cm)[r]{\includegraphics[scale=.9]{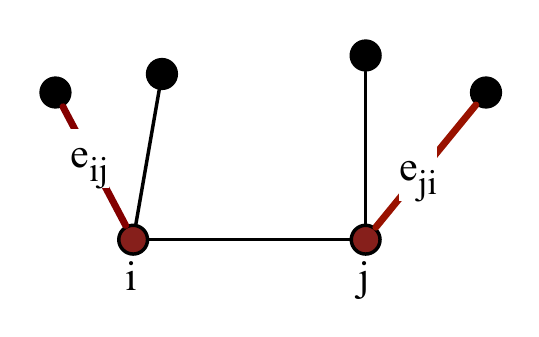}}

  Note that $x'(e)<x(e)$ only if $e$ is incident to
  $j_p$. Furthermore, for such edges $e$, we must have
  \begin{equation}\label{eq:jpedges}
    x(e_{i_pj_p}) +2\mu = x(e_{j_pi_p}) \leq x(e), 
  \end{equation}
  and hence $x'(e)\geq x(e_{i_pj_p}) + \mu$. Similarly, if
  $x'(e)>x(e)$ then $e$ must be incident to $i_p$, and thus
  $x'(e) \geq x(e_{i_pj_p})+\mu$. 

  This has several immediate consequences. For a pair $ij \in S(x)$,
  the $x$-values of $e_{ij}$ $e_{ji}$ remain the same, and they will
  remain surplus-defining edges for $s_{ij}(x')$ and $s_{ji}(x')$,
  respectively. In particular, we have $s_{ij}(x')=s_{ji}(x')$,
  proving (iii). 

  First consider some pair $ij \in I(x)$; by the above reasoning,
  $s_{ij}(x)$ can not increase, but may decrease. If it decreases,
  then it decreases by $\mu$. Similarly, $s_{ji}(x)$ may increase, but
  in this case, $s_{ji}(x') \leq s_{i_pj_p}(x)-\mu$. By (iii), no pair
  $ij \in S(x)$ changes its surpluses, and hence $s(x') \leq
  s(x)$. Furthermore, if $s(x')=s(x)$ then $I(x') \subsetneq I(x)$ as
  the pair $i_pj_p$ certainly decreases its surplus. This implies
  (ii). 

  Finally assume that $x$ satisfies \eqref{stab:1} and
  \eqref{stab:2}. It is clear that the local shift preserves
  \eqref{stab:1}.  Consider some edge $e \in E'$; since \eqref{stab:2}
  holds for $x$, we have $x(e) \geq 1$. Suppose that $x'(e) <
  x(e)$. Then $e$ must be incident to $j_p$, and clearly
  $x(e_{i_pj_p})+2\mu \leq x(e_{j_pi_p}) \leq x(e)$,
  for some $\mu > 0$. Note that $x(e_{i_pj_p}) \geq 1$ as $x$ satisfies
  \eqref{stab:2}, and hence $x(e) \geq 1+2\mu$. This proves (iv) as
  $x(e)$ decreases by no more than $\mu$. \qed
\end{proof}

From (ii) we obtain the following immediate corollary.

\begin{corollary}\label{cor:decr}
  $s(x)$ strictly decreases after at most $|E'|$ local shifts.
\end{corollary}

One could now suspect that a sequence of local shifts would converge
to a prekernel element, but this is in fact not
known. Stearns~\cite{St68} showed, however, that if one picks a {\em
  maximal} violated pair $i_pj_p$ then the method can indeed be shown
to converge, but no polynomial bound on the number of shifts is known. 
Faigle et al.~\cite{FKK98} propose the following elegant fix. Consider
the following linear program with variables $y \in \IR^{E'}_+$ and
$\delta \in \IR$. We let $\Delta(x)$ be the smallest surplus of any
pair in $S(x)$, and let $\Delta(x)=0$ if $S(x)$ is empty.

\begin{align} 
  \max \quad & \delta \label{lp:delta}\tag{P($x$)} \\
    \st \quad & y(V) = \1^Tx \label{delta:1} \\
    & y(e_{ij}) = x(e_{ij}) && \forall ij\in S(x) \label{delta:2} \\
    & 1-y(e_{ij}) \geq 1-y(e) && \forall ij \in S(x), \forall e \in
    \delta(i)\setminus \{ij\} \label{delta:3}\\
    & 1-y(e) \leq \Delta(x) - \delta && \forall ij \not\in S(x),
    \forall e \in \delta(i)\setminus \{ij\} \label{delta:4} \\
    & y(e) \geq 1 && \forall e \in E' \label{delta:5} \\
    & y, \delta \geq 0 \notag
\end{align}
  
Let $x \in \IR^{E'}_+$ satisfy \eqref{stab:1} and \eqref{stab:2}. Then
$y:=x$ and $\delta=\Delta(x)-s(x)$ is a feasible solution to
\eqref{lp:delta}. Faigle et al. show a stronger statement.  Consider a
sequence
$$ x=x^1, x^2, \ldots, x^p $$
of edge vectors such that for all $2 \leq i \leq p$, $x^i$ arises from
$x^{i-1}$ through a local shift operation. Note that Lemma
\ref{lem:maschler} implies that $S(x^1) \subseteq S(x^2) \subseteq
\ldots \subseteq S(x^p)$. 

\begin{lemma}[\cite{FKK98}] \label{lem:faigle} 
  Suppose that $S(x^p)=S(x^1)$. For $1 \leq i \leq p$, define
  $y^i=x^i$ and let $\delta^i=\Delta(x^i)-s(x^i)$. Then
  $(y^i,\delta^i)$ is feasible for \eqref{lp:delta}, and $\delta^1
  \leq \delta^2 \leq \ldots \leq \delta^p$. 
\end{lemma}
\begin{proof}
  We prove the lemma by induction of $i$. For $i=1$, we have already
  shown that $(y^1,\delta^1)$ is feasible for \eqref{lp:delta}.  Now
  assume that the induction hypothesis is true for $i \geq 1$. We show
  that $(y^{i+1},\delta^{i+1})$ is feasible.  It is easy to see that
  the local shift preserves constraints \eqref{delta:1} and
  \eqref{delta:2}. Lemma \ref{lem:maschler}.(iii) implies that
  \eqref{delta:3} continues to hold. The left-hand side of
  \eqref{delta:4} is at most $s(x^{i+1})$, and hence at most 
  $\Delta(x^{i+1})-\delta^{i+1}$. 
  Lemma \ref{lem:maschler}.(ii) together with the fact that
  $\Delta(x^i)=\Delta(x^{i+1})$ implies that $\delta^{i} \leq
  \delta^{i+1}$.  Finally, the core condition \eqref{delta:5}
  holds because of Lemma \ref{lem:maschler}.(iv). \qed
\end{proof}

The idea of Faigle et al. is now as follows: let $x$ be an edge-vector
that satisfies \eqref{stab:1} and \eqref{stab:2}. Solve
\eqref{lp:delta}, and let $(y,\delta)$ be the optimal
solution. We have that $s(y) \leq s(x)$, 
$$ s_{ij}(y)=s_{ji}(y), $$
for all $ij \in S(x)$, and therefore $S(x) \subseteq S(y)$. If this
inclusion is strict, then we have made progress by solving the
LP. Assume that $S(x)=S(y)$. In this case, we observe that the maximum
value of $1-y(e)$ over all left-hand sides of \eqref{delta:4} is
exactly $s(y)$, and we therefore have
$$ \delta=\Delta(x)- s(y). $$
Apply local moves to $y$ until $s(y)$ decreases. We know
from Corollary \ref{cor:decr} that this takes at most $|E'|$
steps. Let $y'$ be the resulting vector. Clearly, $S(y') \subseteq
S(y)$, and we claim that this inclusion must be strict. In fact, if
not then Lemma \ref{lem:faigle} implies that $(y',\delta')$ is
feasible for \eqref{lp:delta} for 
$$ \delta' = \Delta(x)-s(y') > \Delta(x)-s(y) = \delta, $$ 
where the inequality follows from that fact that $s(y') < s(y)$.  This
contradicts the optimality of $(y,\delta)$ for \eqref{lp:delta}.

\begin{theorem}[\cite{FKK98}] Given a point $x \in \IR^{E'}_+$ that
  satisfies \eqref{stab:1} and \eqref{stab:2} we can compute a point
  in the prekernel using $|E'|^2$ local shifts, and solving $|E'|$ LPs
  of the type \eqref{lp:delta}. 
\end{theorem}

\section{Discussion}

%%%%%%%%%%%%%%%%%%

In this paper we studied network bargaining games, and took a particular interest in unstable games. We showed that when the underlying network, $G=(V,E)$,  is sparse, we are able to identify an approximation of the smallest  blocking set,  $B$,  and we explained how it is possible to efficiently find a balanced (and thus stable) outcome in the network game induced by the graph, $G'=(V,E\setminus B)$.
There are several interesting directions in which this work could be taken. For example, one could allow for more complex utility functions~\cite{CK08} or allow agents to make more complex deals.

We make the observation that in games defined by $G$ and $G'$ we are
always working with maximum matchings, $M$ and $M'$ on the respective
networks.  Note that the matching $M'$ is no larger, and may in fact
be smaller than matching $M$. Therefore, we achieve balancedness by
{\em subsidizing} the game to an extent of $|M|-|M'|$. This difference
is clearly at most $|B|$, and hence it makes sense to use as small a
blocking set as we can find.
%This subsidy perspective brings to mind the literature on the \emph{cost of stability}~\cite{BE+09}.
Bachrach et al looked at stabilizing coalitions (in non-network
settings) by using external payments, and introduced the \emph{cost of
  stability} as the minimal external payment needed to stabilize a
game~\cite{BE+09}.
 % Both our and their approaches are trying to delve into the
 % underlying causes of instability, and provide ideas as to how to
 % talk about how far a game is from being stable.
 An interesting future direction might be to explore relationships and possible tradeoffs between the cost of stability and the size of the blocking set.
}{}

\bibliographystyle{splncs03}
\bibliography{nwbg}

\newpage

\iffalse
\section{Appendix A}
 
\subsection*{Proof of Lemma~\ref{lem:maschler}}
\fi

\end{document}